\documentclass[12pt]{article}
\usepackage{amsfonts, amssymb, amsmath, amsthm}
\usepackage[margin=1.25in]{geometry}

\usepackage[round, longnamesfirst]{natbib}
\usepackage{multirow}
\usepackage{comment}
\usepackage[english]{babel} 
\usepackage{amsmath,amsfonts,amsthm,bm} 
\usepackage{array}
\newcolumntype{?}{!{\vrule width 2pt}}
\newlength{\Oldarrayrulewidth}

\usepackage[justification=centering,textfont={sc},labelfont={sc}]{caption,subcaption}
\newcommand{\interior}[1]{%
  {\kern0pt#1}^{\mathrm{o}}%
}
\newcommand{\costpar}{r}

\newcommand{\cbar}{\costpar^{\dagger}}

\newcommand{\gX}{\mathcal{X}}

\usepackage[capposition=bottom]{floatrow}
\newtheorem{theorem}{Theorem}

\newtheorem{corollary}{Corollary}

\newtheorem{definition}{Definition}
\newtheorem{lemma}{Lemma}
\newtheorem{proposition}{Proposition}

\usepackage{apptools}
\AtAppendix{
  \counterwithin{lemma}{section}
  \counterwithin{remark}{section}
  \counterwithin{corollary}{section}
  \counterwithin{definition}{section}
  \counterwithin{proposition}{section}
}
\NewDocumentEnvironment{example}{ o }{
  \par\noindent\textbf{Example%
    \IfValueT{#1}{\space #1}%
  .}\quad\itshape
}{\par}

\usepackage{pgfplots}
\usepgfplotslibrary{fillbetween}
\usetikzlibrary{intersections}
\pgfplotsset{compat=1.13}
\makeatletter
\renewenvironment{proof}[1][\proofname]{\par
  \pushQED{$\blacksquare$}%
  \normalfont \topsep6\p@\@plus6\p@\relax
  \trivlist
  \item[\hskip\labelsep
        \bfseries
    #1\@addpunct{:}]\ignorespaces
}{%
\hfill \popQED\endtrivlist\@endpefalse
}
\makeatother

\renewcommand{\interior}[1]{\text{int}({#1})}
\usepackage[hidelinks]{hyperref}
\usepackage{cleveref}
\crefname{subsection}{subsection}{subsections}
\Crefname{subsection}{Subsection}{Subsections}
\usepackage{enumerate}
\usepackage{appendix}

\title{The Limits of Limited Commitment\thanks{%
We thank the editor, Alexander Wolitzky, as well as  Giacomo Calzolari, Matteo Escud\'e, Dino Gerardi, Francesco Giovannoni, Aditya Kuvalekar, Ignacio Monz\'on, Bal\'azs Szentes and three anonymous referees  for their insightful comments.}
}
\author{
Jacopo Bizzotto
\thanks{
OsloMet, \url{jacopo.bizzotto@oslomet.no}.}
\and
Toomas Hinnosaar%
\thanks{University of Nottingham and CEPR, \url{toomas@hinnosaar.net}.}
\and
Adrien Vigier%
\thanks{University of Nottingham, \url{adrien.vigier@nottingham.ac.uk}.}
}
\date{\today}

\begin{document}

\maketitle

\begin{abstract} 
We study limited strategic leadership.  A collection of subsets covering the leader's action space determines her commitment opportunities. We characterize the outcomes resulting from all possible commitment structures of this kind. If the commitment structure is an interval partition, then the leader's payoff is bounded by her Stackelberg and Cournot payoffs. Under general commitment structures the leader may obtain a payoff that is less than her lowest Cournot payoff. We apply our results to a textbook duopoly model and elicit the   commitment structures  leading to consumer- and producer-optimal outcomes.
\end{abstract}



\section{Introduction}

In many strategic settings, a
player's early decisions constrain the actions  she will later be able to take.
The 
Stackelberg leadership model provides  the best-known  example of this kind: one firm (the leader) chooses a quantity first; another firm (the follower) observes the decision of the leader, and then chooses  its own quantity. Numerous
variants of this model are widely used,  in 
industrial organization \citep{tirole_theory_1988}, macroeconomics \citep{kydland_rules_1977,barro_rules_1983},  computer science, operations research, and marketing. Yet, 
by allowing the leader to select a \emph{single} action early on, the Stackelberg model supposes  a very strong 
form of commitment. While 
this ``full-commitment" assumption is 
undoubtedly warranted in a number of settings, 
 many real-world  situations are such that  early decisions restrict---but do not strictly fix---a player's subsequent choice. The term ``limited commitment" captures   the leader's inability to commit to a single  action of her choice.

Limited commitment can manifest itself in many ways.
For instance, an incumbent firm could influence potential entrants by investing in the development of a new product, building a plant in a new country, or announcing a list price while leaving the door open to discounts. We thus propose a model which allows for arbitrary forms of limited commitment. We characterize all outcomes resulting from limited commitment, compare them to the full and no-commitment benchmarks, and thereby evaluate the robustness of predictions based on these benchmarks.

The following example illustrates the power of limited commitment.
There are two firms, an incumbent and an entrant, respectively producing quantities $q_1$ and $q_2$;  production costs are zero, and the inverse demand function is $1-q_1-q_2$. The full-commitment (or Stackelberg) outcome is such that the incumbent produces $1/2$ whereas the entrant produces $1/4$;  the resulting price is $1/4$.  By contrast, the no-commitment (or Cournot) outcome is such that each firm produces $1/3$, leading to  the  higher price   $1/3$. As it turns out, limited commitment may benefit consumers by inducing a price  lower than  in both  of these cases. For instance, suppose that the only  commitment options  of the incumbent are (i) produce at most $2/3$, (ii) produce at least $2/3$.
In that case, the incumbent chooses the second option and ends up producing $2/3$. The entrant produces $1/6$, and the resulting  price is $P=1/6$.

The previous example is a special case of the general framework we study. In our model, there are two periods and two players, a leader and a follower. Players' action spaces are closed intervals and their payoffs are quasi-concave in one's own actions. 
The setting is parametrized by a collection of subsets that cover the leader's action space; we refer to this collection of subsets as the commitment structure. In the first period, the leader selects an element from the  commitment structure. In the second period,  leader and follower simultaneously choose one action each, the leader being restricted to pick an action from the subset that she selected in the first period. The full- and no-commitment benchmarks are special cases of our model: in the former, the  commitment structure  consists of singletons; in the latter, the  commitment structure  includes only one subset---the leader's entire action space.

We characterize the set of outcomes that are part of an equilibrium under some  commitment structure.
We label those outcomes  \emph{plausible}.
As illustrated in the example above, the set of plausible outcomes extends beyond the bounds defined by the full and no-commitment benchmarks.\footnote{That is, players' equilibrium actions need not be bounded by their full and no-commitment counterparts.} 
More strikingly, this remark applies to the players' equilibrium payoffs too. 
In particular, although full commitment always benefits the leader, we show that the leader can be worse off under limited commitment than under no commitment. In fact, in a standard symmetric duopoly setting---where full commitment confers a well-known first-mover advantage---certain forms of limited commitment even generate a second-mover advantage.

We then examine the set of outcomes resulting from ``simple'' commitment structures, namely commitment structures  which divide the leader's action space into non-overlapping intervals. We call these outcomes \emph{simply plausible}. We show that any outcome such that the follower best-responds and guaranteeing the leader at least her highest no-commitment payoff is simply plausible. Furthermore,  any simply-plausible outcome must give the leader at least  her lowest no-commitment payoff. Thus limited commitment always benefits the leader  under simple commitment structures.
However, the set of simply-plausible outcomes might exhibit ``gaps,'' in the sense that certain outcomes giving the leader a payoff between her full and no-commitment payoffs may fail to be simply plausible.\footnote{The function $U$ mapping every action of the leader to the payoff she obtains when the follower best-responds plays a central role in our results. A subset of the leader's action space is called \emph{$U$-monotone} if whenever an action is contained in it so is the upper contour set of that action with respect to $U$. We show that the plausible actions of the leader are $U$-monotone, whereas the simply-plausible ones need not be.}

The final section of our paper presents an application of our results to a textbook duopoly setting. We first characterize all plausible outcomes of that setting, and then describe the consumer surplus and social welfare-maximizing plausible outcomes. We show that, in terms of consumer surplus and total welfare, some forms of limited commitment are always strictly preferable to the traditional benchmarks of full and no commitment.

We contribute to a literature on commitment whose starting point is that economic agents often commit to \emph{subsets} of actions rather than single actions. This literature is divided in two branches. One branch  posits that agents commit to subsets of actions because they are constrained to do so: even if they wanted to, agents would be unable to commit to specific actions. This  includes \citet{spence_entry_1977}, where an incumbent firm faces a prospective entrant and invests in productive capacity during the first period but may choose not to utilize the full capacity during the second period. It also includes  \citet{saloner_cournot_1987}, \citet{admati_joint_1991}, and \citet{romano_endogeneity_2005}, where agents can set a lower bound on the action they will choose but retain until the last period the option to choose any action that is at least as large as this lower bound. The present  paper belongs to this branch of literature; our contribution is to consider \emph{all} possible commitment structures, that is, all collections of subsets that the leader may choose from in the first period.

A second branch of literature proposes that agents commit to subsets of actions because doing so gives them a strategic advantage. In these models, agents could commit to a  single action but typically \emph{choose} to commit to  a subset containing more than one action because doing so enables them to credibly threaten potential deviating players. \cite{bade_bilateral_2009}, \cite{renou_commitment_2009}, and \cite{dutta_dynamic_2016}, embed a strategic-form game into a multi-stage game in which, early on in the game,  players can freely restrict their action spaces.\footnote{
For a multi-stage model of commitment where the base-game itself is an extensive-form game, see \cite{arieli_sequential_2017}.
} 
A player can commit to any action of his choice, but may also choose not to commit at all. This ability to freely choose what to commit to differs from our model and has stark implications: in our setting, if the leader can choose not to commit at all, she can guarantee herself her lowest Cournot payoff, while if she can commit to any action, then she can guarantee herself her Stackelberg payoff. By contrast, we show that if the leader's ability to commit is limited, she will  not only fall short of her Stackelberg payoff but may also fall short of her lowest Cournot payoff.\footnote{Relatedly, \cite{pei_when_2016} studies a setting where a player can restrict the actions of her opponent. \cite{pei_when_2016} shows that when a player has limited commitment ability, i.e., cannot reduce her opponent's action set to a singleton, then it is sometimes strictly optimal not to restrict the opponent's actions at all.}
Notwithstanding, Proposition 1 in
\cite{bade_bilateral_2009} 
constitutes a key building block towards our  characterization of the set of simply plausible outcomes.
Note that 
several other papers study what might be construed  as a form of  limited commitment:
some allow agents to pick specific actions but let them revise these choices later on, either at fixed times \citep{maskin_theory_1988}, stochastically \citep{kamada_revision_2020}, or by incurring various costs \citep{henkel_1.5th_2002,caruana_production_2008}.\footnote{
For a study of more general commitment devices that include commitments contingent on other players' commitments, see \cite{kalai_commitment_2010}.
}

Finally,  our paper belongs to a recent strand of papers that take a base game as given and examine how changing the structure of this game can affect its outcome. For example,  \cite{kamenica_bayesian_2011}, \cite{bergemann_bayes_2016},  and \cite{makris_information_2023}
examine the implications of changing a game's information structure. \cite{nishihara_resolution_1997} and \cite{gallice_co-operation_2019} study instead the effects of changing the order of moves. \cite{salcedo_implementation_2017} and \cite{doval_sequential_2020} allow the structure of the game to change in both of these dimensions.

\section{The Model} \label{sec:model}

\subsection{Setup}

There are two players, a \textit{leader} and a \textit{follower}, with action spaces $\mathcal{X}=[\underline{x}, \overline{x}]$ and $\mathcal{Y}=[\underline{y}, \overline{y}]$, respectively.  A collection $K$ of non-empty subsets of $\mathcal{X}$ covers the leader's action space.\footnote{That is, (i) for all $\mathcal{X}_i \in K$:   $\mathcal{X}_i \subseteq \mathcal{X}$; (ii) for all $x \in \mathcal{X}$: $x \in \mathcal{X}_i$ for some $\mathcal{X}_i \in K$.
}
We refer to $K$ as the \textit{commitment structure} (CST).

There are two periods: in period 1, the leader publicly selects $\mathcal{X}_i \in K$; in period 2, leader and follower simultaneously choose actions $x$ and $y$, with $x$ contained in $\mathcal{X}_i$ and $y$ contained in $\mathcal{Y}$.
The resulting payoffs are $u(x, y)$ for the leader and $v(y,x)$ for the follower, where $u$ and $v$ are continuous. We further assume that $u(x,y)$ is strictly quasi-concave in $x$ for all $y \in \mathcal{Y}$, and that $v(y,x)$ is strictly quasi-concave in $y$ for all $x \in \mathcal{X}$.
This game is denoted by $G(K)$.
Throughout, our focus is on  pure-strategy subgame perfect equilibrium, henceforth simply referred to as \textit{equilibrium}.

\subsection{Definitions and Notation} \label{subsec:notation}

An action pair $(x,y)$ with $x \in \mathcal{X} $ and $y \in \mathcal{Y}$ is referred to as an \textit{outcome}. 
We say that an outcome $(x,y)$ is \textit{plausible} if $(x,y)$ is an equilibrium outcome of $G(K)$ for some commitment structure $K$. An action $x$ is \textit{plausible} if it is part of a plausible outcome $(x,y)$.

Two salient commitment structures play a central role, 
\[
K^S:=\big\{ \{x\}: x \in \mathcal{X} \big\} ~~ \text{and} ~~
K^C:=\big\{ \mathcal{X}\big\};
\]
we refer to these as the Stackelberg and Cournot CSTs, respectively. By extension, the 
equilibrium outcomes of $G(K^S)$ and $G(K^C)$ will be referred to as Stackelberg and Cournot outcomes. The Cournot actions of the leader are the actions of the leader forming part of a Cournot outcome.

A commitment structure $K$ is said to be \textit{simple} if it partitions the leader's action space into intervals. For example, 
the Stackelberg and Cournot CSTs are  simple CSTs. An outcome $(x, y)$ is \textit{simply plausible} if $(x, y)$ is 
an equilibrium outcome of $G(K)$, for some simple commitment structure $K$.

To every action $x$ of the leader corresponds a unique best response of the follower.\footnote{Recall, the follower's action space is compact, and  $v(y,x)$  is strictly quasi-concave in $y$.} We denote this best response  by $R_F(x)$, and let $U(x)$ be the payoff of the leader from taking action $x$ when the follower best-responds to $x$, that is, 
\[
  U(x):=u\big(x, R_F(x)\big). 
\]
A subset $\tilde{\mathcal{X}}\subseteq \mathcal{X}$ is called \emph{$U$-monotone} if $\tilde{x}\in \tilde{\mathcal{X}}$ implies that the upper contour set of $\tilde{x}$ with respect to $U$ is contained in $\tilde{\mathcal{X}}$ as well.\footnote{The  upper contour set of an action  $x$ with respect to $U$ is the set of actions  $\tilde{x}$ such that  $U(\tilde{x}) \geq U(x)$.}

\subsection{Duopoly Example}
\label{subsec:duoexample}

In this subsection, we illustrate the model in a textbook duopoly setting. Leader and follower are two identical firms, each choosing a quantity in $\mathcal{X}=\mathcal{Y}=\big[0,2/(2-\costpar)\big]$.\footnote{Quantities larger than $2/(2-\costpar)$ would lead to negative profits no matter what.} A firm producing quantity $q$  incurs  cost $3q-\costpar q^2/2$, where $\costpar<2$ measures the returns to scale. Firms sell at unit price $4-(x+y)$.  
Letting $u(x,y)$ (respectively, $v(y,x)$) be the profit of the leader (respectively, the follower)  gives  $v(y,x)=u(y,x)$ and
\begin{equation}\label{eqduopolypayoffs}
  u(x,y)=x-xy-\Big(1-\frac{\costpar}{2}\Big)x^2.
\end{equation}

We set for now $\costpar=0$. The (unique) Cournot and Stackelberg actions are then, respectively, $x^C=1/3$ and $x^S=1/2$. 
Let us consider the simple commitment structure
\[
K=\left\{\left[0,\frac{3}{5}\right),\left[\frac{3}{5}, 1\right]\right\}.
\]
\Cref{F:example1} illustrates this example. Any quantity in the interval $[3/5,1]$ is such that, whenever the follower best-responds, the leader benefits from deviating to a smaller quantity. Hence, any equilibrium of $G(K)$ must be such that the leader produces $3/5$ in the corresponding subgame. If instead the leader picks $[0,3/5)$ in the first period, then each firm produces the Cournot quantity. As $U(3/5)>U(x^C)$, the unique equilibrium is such that in period 1 the leader chooses the upper interval. 

\begin{figure}
  \centering
  \includegraphics[trim={6pt 6pt 12pt 0pt},clip,width=0.55\linewidth]{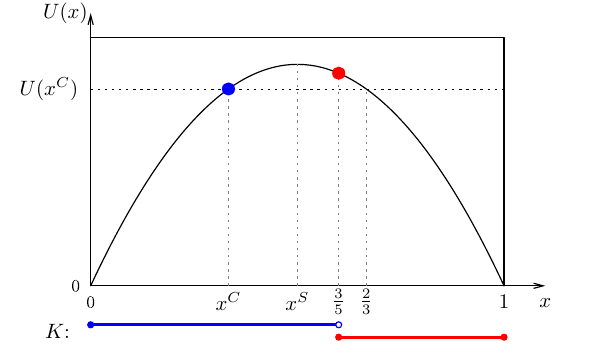}
  \caption{} \label{F:example1}  
\end{figure}

The previous reasoning applies if $3/5$ is replaced by any quantity $x^*$ in the interval $(1/3,2/3]$,  so all actions in $[1/3,2/3]$ are simply plausible. We will see  below that  these are in fact the only simply-plausible actions:
we show in 
 \Cref{sec: mainone} that  any  simply-plausible outcome is such that the
  leader's payoff is  bounded by her Stackelberg
and Cournot payoffs. 
 However, we  show  in \Cref{sec: maintwo} that the set of plausible outcomes may   be larger than the set of simply-plausible ones; moreover, in general,  the Cournot payoffs  do not bound the payoffs that the leader can obtain.

\section{Preliminaries}
Henceforth, 
let
\[
\eta( \tilde{x},x):=u\big(\tilde{x}, R_F(x)\big)- u\big(x, R_F(x)\big). 
\]
In words, $\eta( \tilde{x},x)$ measures the leader's gain from choosing $\tilde{x}$ instead of $x$ when the follower best-responds to $x$. 

Consider an arbitrary commitment structure $K$. Suppose that 
an equilibrium of $G(K)$ exists. Given $\mathcal{X}_i \in K$,  write $\beta(\mathcal{X}_i)$ for the leader's action in the subgame following $\mathcal{X}_i$. Then $\beta(\mathcal{X}_i) \in \mathcal{X}_i$, and $\eta\big(x, \beta(\mathcal{X}_i)\big) \leq 0$ for all $x \in \mathcal{X}_i$. The notion of \textit{admissible pair} summarizes these basic properties.

\begin{definition}
A pair $(K,\beta)$ made up of a commitment structure $K$ and a mapping $\beta: K \rightarrow \mathcal{X}$ is said to be \emph{admissible} if
\begin{enumerate}
\item [(a)] $\beta(\mathcal{X}_i) \in \mathcal{X}_i$, for all $\mathcal{X}_i \in K;$
\item [(b)] $\eta\big(x, \beta(\mathcal{X}_i)\big) \leq 0$,  for all   $x \in \mathcal{X}_i$ and all $\mathcal{X}_i \in K$.
\end{enumerate}
\end{definition}

The following characterization 
 of the set of plausible outcomes is immediate.

\begin{lemma}\label{rem:remark1} 
An outcome $(x, y)$ is plausible if and only if there exist an admissible pair $(K,\beta)$ and $\mathcal{X}_{i}\in K$, such that 
\begin{enumerate}
\item[(i)] $x= \beta(\mathcal{X}_{i} )$,
\item[(ii)] $U(x ) =\max_{\mathcal{X}_j \in K} U\big(\beta(\mathcal{X}_j)  \big)$,
\item[(iii)] $y=R_F(x)$.
\end{enumerate}
\end{lemma}

\noindent
We say that an admissible pair $(K,\beta)$ \emph{implements} outcome $(x, y)$ if it satisfies conditions (i)--(iii) of Lemma \ref{rem:remark1}. We can then rephrase the lemma to say that an outcome $(x, y)$ is plausible if and only if some admissible pair $(K,\beta)$ implements it.

 \section{Simple Commitment Structures } \label{sec: mainone}  
  
This section contains the first part of our  analysis; in it,  we characterize the set of simply-plausible outcomes. All proofs for this section are in  \Cref{A:intervals}.

Denote by $R_L(y)$ the unique best response of the leader to the follower's action $y$, and define\footnote{The leader's action space is compact and $u(x,y)$ is strictly quasi-concave in $x$.
} 
\[
\phi(x):=R_L\big(R_F(x)\big). 
\]
 The fixed points of $\phi$ are thus the Cournot actions of the leader.
 Let $\mathcal{X}^C$ denote said  set of  Cournot actions; the notation $x^C_n$ will indicate a generic element of this set.


The following lemma states a central   property of admissible pairs $(K,\beta)$ such that $K$ is a CST comprising only intervals. The  essence of this  lemma is akin to Proposition 1 in \cite{bade_more_2007}. 

\begin{lemma}\label{lem:extension}
Let $K$ be a  commitment structure comprising only intervals. Then $(K,\beta)$ is admissible if and only if, for all $\mathcal{X}_i\in K$, one of the following conditions holds: 
\begin{enumerate} 
\item [(i)] $\beta(\mathcal{X}_i)\in \mathcal{X}_i \cap \mathcal{X}^C$; 
\item [(ii)] $\beta(\mathcal{X}_i)=\min \mathcal{X}_i$ and $ \phi\big(\beta(\mathcal{X}_i)\big)<\beta(\mathcal{X}_i)$; 
\item [(iii)] $\beta(\mathcal{X}_i)=\max \mathcal{X}_i$ and $\phi\big(\beta(\mathcal{X}_i)\big)>\beta(\mathcal{X}_i)$. 
\end{enumerate}
\end{lemma}

The intuition behind the lemma is straightforward. Consider an interval 
$\mathcal{X}_i$ forming part of a commitment structure $K$, and  a mapping $\beta: K \rightarrow \mathcal{X}$ such that $\beta(\mathcal{X}_i) \in \mathcal{X}_i$ for all $\mathcal{X}_i \in K$. 
Suppose that 
$\phi\big(\beta(\mathcal{X}_i)\big)>\beta(\mathcal{X}_i)$. Then, when the follower best-responds to $\beta(\mathcal{X}_i$), the
leader would like to play an action slightly  greater than $\beta(\mathcal{X}_i)$.
For 
 $(K,\beta)$ to be  admissible  the leader must therefore be  unable to   play such a  greater  action after choosing $\mathcal{X}_i$ in period 1. Since  $\mathcal{X}_i$ is an interval, this amounts to saying that   $\beta(\mathcal{X}_i)$ is  the upper bound of  $\mathcal{X}_i$.

 \Cref{F:fig1}, panel A, illustrates Lemma \ref{lem:extension} in the context of the duopoly example introduced in Subsection \ref{subsec:duoexample}, for  $\costpar=6/5$. The black curve represents the graph of the function $\phi$. 
 The leader's Cournot actions are $x_1^C=0$, $x_2^C=5/9$, and $x_3^C=5/4$. An admissible pair $(K,\beta)$ is such that every action $\beta(\mathcal{X}_i)$ belonging to a region of the figure with a left-pointing arrow (respectively, right-pointing arrow) is either a Cournot action or the leftmost (respectively, rightmost) element of $\mathcal{X}_i$.

\begin{figure}[!htb]
  \centering
  \subfloat[]{\includegraphics[width=0.47\textwidth,trim={12pt 8pt 21pt 8pt},clip]{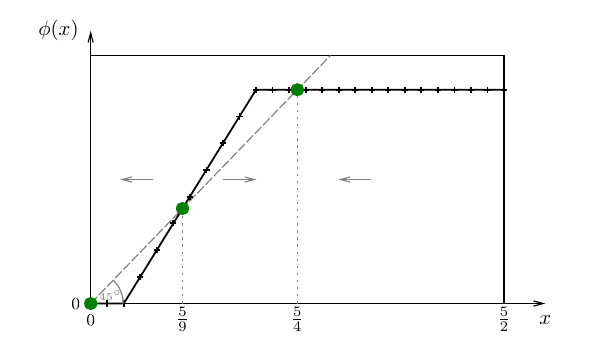}\label{F:phi}}
  \hfill 
  \subfloat[]{\includegraphics[width=0.47\textwidth,trim={12pt 8pt 21pt 8pt},clip]{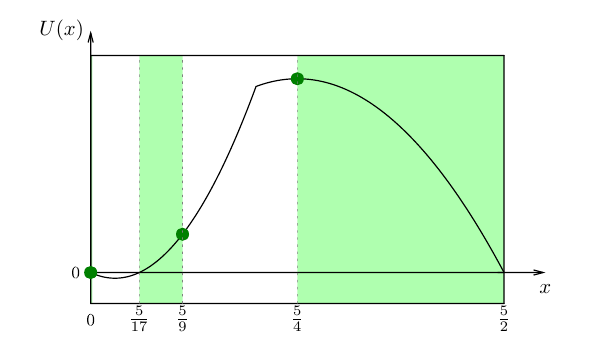}\label{F:U}}
  \caption{}
  \label{F:fig1}
\end{figure}

Our first theorem characterizes the set of simply-plausible outcomes.

\begin{theorem}\label{thm3}
An action $x^*$ is simply plausible if and only if the lower contour set of $x^*$ with respect to $U$ contains  a Cournot action $x_{n^*}^C$ such that
\begin{equation}\label{eq:thm1}
\big(\phi(x^*)-x^*\big)\big(x_{n^*}^C- x^*\big)\geq 0.
\end{equation}
\end{theorem}

\Cref{thm3} tells us that an action $x^*$ at which $\phi(x^*)>x^*$ is simply plausible if and only if some Cournot action greater than $x^*$ belongs to the lower contour set of $x^*$ with respect to $U$. The \emph{if} part is easy. Let $x_{n}^C$ be a Cournot action greater than $x^*$ and suppose that it belongs to the lower contour set of $x^*$. Now consider $K=\big\{ [\underline{x}, x^*], (x^*, \overline{x}] \big\}$, and $\beta$ given by $\beta\big( [\underline{x}, x^*] \big)=x^*$ and $\beta\big( (x^*, \overline{x}] \big)=x_{n}^C$.  By  \Cref{lem:extension},
the pair $(K,\beta)$ is admissible; it implements $x^*$ since $U(x^*)\geq U(x_{n}^C)$.\footnote{In particular, it follows  that every simply-plausible outcome is an equilibrium outcome of a game $G(K)$ for some simple commitment structure $K$ comprising   two elements only.}


The proof of the \textsl{only if} part  goes as follows.
Pick a simply plausible action $x^*$ such that $\phi(x^*)>x^*$. Reason  by contradiction, and assume that all Cournot actions of the leader greater than $x^*$ are in the strict upper contour set of $x^*$. We first show that some Cournot action of the leader has to be greater than $x^*$, and let $x_1^C$ denote the smallest such action. 
Note that  $x_1^C$ cannot be an equilibrium action in the corresponding subgame; something else is, say $x_1$. 
We argue that $x_1$ must be greater than $x_1^C$ and such that $\phi(x_1)>x_1$.
We now repeat the argument above with $x_1$ in place of $x^*$, and so on.
 We show that  the previous recursion must end after finitely many steps, say $n$, for otherwise 
 the fact that the action spaces are compact and the  payoff functions are continuous leads to a contradiction. Yet $\phi(x_n)>x_n$ implies $x_n<\overline{x}$, giving another  contradiction.

Applying \Cref{thm3} to the example of 
\Cref{F:fig1} shows that the set of simply-plausible actions is equal to $ \{0\} \cup \big[ 5/17, 5/9 \big] \cup \big[5/4,5/2 \big]$ (illustrated in green in panel B of \Cref{F:fig1}). Firstly, \Cref{thm3} shows that no action in the interval $\big(0,5/17 \big)$ is simply plausible, since all of them belong to the strict lower contour set of each Cournot action. Secondly, any $x \in \big(5/9,5/4\big)$ satisfies $\phi(x) >x$ (see panel A). The only Cournot action greater than any of these actions is $x_3^C$. As $U(x_3^C)>U(x)$ for all $x \in \big(5/9, 5/4 \big)$, we conclude using \Cref{thm3} that no action in this interval is simply plausible. Mirror arguments show that all actions in $ \{0\} \cup \big[ 5/17, 5/9 \big] \cup \big[5/4, 5/2 \big]$ are simply plausible.\footnote{An action $x^*\in [ 5/17, 5/9] \cup [5/4,5/2]$ is implemented, for instance, by the pair $(K,\beta)$ where $K=\left\{[0,x^*),[x^*,5/2]\right\}$, $\beta\big([0,x^*)\big)=0$, and $\beta\big([x^*,5/2]\big)=x^*$.}

By construction, the leader's Stackelberg payoff provides an upper bound for the payoffs attainable by the leader under any CST. \Cref{thm3} shows that the Cournot payoffs provide a corresponding lower bound
for simple CSTs. Moreover, 
\Cref{thm3} implies that 
 if an action is in the upper contour set of \emph{all} Cournot actions, then that action must be simply plausible. The following corollary records these observations.

\begin{corollary}\label{cor:UCS_LCS}
All simply-plausible actions belong to the upper contour set of a Cournot action with respect to $U$. Furthermore, any action in the intersection of these upper contour sets is simply plausible. 
\end{corollary}

\section{Beyond Simple  Commitment Structures} \label{sec: maintwo} 

We saw in the previous section that all simply-plausible outcomes guarantee the leader at least her lowest Cournot payoff. The following examples show that the conditions imposed on simple CSTs are crucial: in general, an outcome may be plausible and give the leader a payoff that is less than her minimum Cournot payoff.

\begin{example}[A]
Consider the duopoly example introduced in Subsection \ref{subsec:duoexample}, with  $\costpar=4/5$
and commitment structure 
\[
K=\left\{\left(\frac{1}{8},\frac{1}{3}\right],\left[0,\frac{1}{8}\right]\cup\left(\frac{1}{3}, \frac{5}{3}\right]\right\}.
\]
As one of its elements is not an interval, $K$ is not a simple CST.
\Cref{F:example2} illustrates this example. The subgame following the leader's choice of $(1/8,1/3]$ possesses a unique equilibrium, in which the leader produces $1/3$. The other subgame has two equilibria: one yielding the Cournot outcome,  $x^C=5/11$, the other involving the leader choosing quantity $1/8$. As $U(1/8)<U(1/3)<U(x^C)$, we see that $G(K)$ possesses two 
equilibria: one in which the leader produces $1/3$, and one in which the leader produces $x^C$. In the former equilibrium, the leader anticipates that if she were to select $\left[0,1/8\right]\cup\left(1/3, 5/3\right]$ in period 1, the follower would respond by producing a quantity larger than $x^C$. Consequently, the leader settles for the quantity $1/3$, and obtains less than the Cournot payoff $U(x^C)$.

\begin{figure}
	\centering
	\includegraphics[trim={12pt 0pt 21pt 0pt},clip,width=0.55\linewidth]{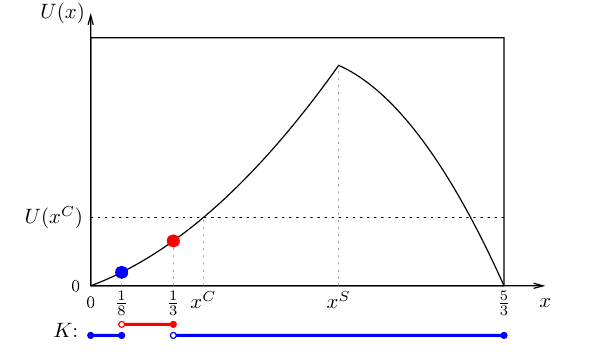}
	\caption{ } \label{F:example2}
\end{figure}
\end{example}

\begin{example}[B]
Consider the following setting.
The action spaces are $\mathcal{X}=\mathcal{Y}=[0,1]$. The payoffs of the leader are given by
\[
 u(x,y)
 = x y + (1-x)(1-y) - \frac{1}{2} \left(x-\frac{1}{2} \right)^2 - \frac{3}{2} \left(y-\frac{1}{2} \right)^2;
\]
the payoffs of the follower are given by $v(y,x)=u(y,x)$. 
In this setting, 
the leader's Stackelberg actions are $0$, $1/2$, and $1$, and  these are also the leader's Cournot actions.
The  commitment structure is 
\[
K=\big\{ [0,x^*], [1-x^*,1] \big\}, 
\]
where $x^*$ denotes some action in $(1/2,1)$.
As $x^*>1-x^*$, $K$ does not partition the leader's action space; so $K$ is not a simple CST. \Cref{F:example3} illustrates this example. The subgame induced by the leader's choice of $ [0,x^*]$ has an equilibrium in which the leader chooses the action  $x^*$. Symmetrically, the subgame induced by the leader's choice of $ [1-x^*,1]$ has an equilibrium in which the leader chooses the action $1-x^*$. Therefore, since $U(x^*)=U(1-x^*)$, both $x^*$ and $1-x^*$ are plausible. In both cases, the leader obtains less than her Cournot payoff $U(1/2)$.

\begin{figure}
\centering
\includegraphics[trim={12pt 0pt 21pt 0pt},clip,width=0.55\linewidth]{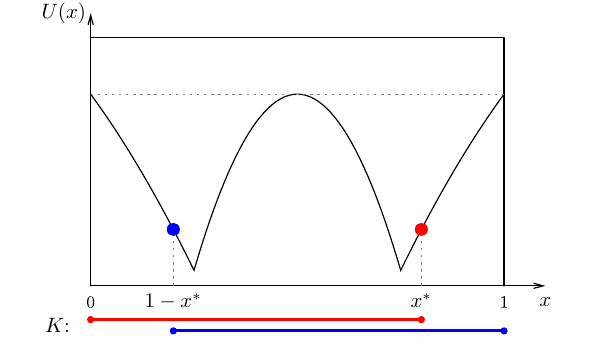}
\caption{} \label{F:example3}
\end{figure}    
\end{example}

A simple CST has two defining properties: it partitions the leader's action space and consists solely of intervals. 
The previous examples illustrate that relaxing either of these conditions can expand the set of outcomes that are plausible. In particular, relaxing either of these conditions can induce the leader to obtain less than her minimum Cournot payoff. \Cref{subsecIplaus} characterizes the set of outcomes induced by CSTs comprising only intervals. CSTs comprising non-convex sets are examined in \Cref{subsecPplaus}: we characterize the set of outcomes induced by CSTs which partition the leader's action space and show that every plausible outcome is plausible under such a  CST. 

\subsection{Commitment Structures Comprising Only Intervals } \label{subsecIplaus}  
  
We say that an outcome $(x,y)$ is \textit{I-plausible} if it is 
an equilibrium outcome of $G(K)$ for some commitment structure $K$ comprising only intervals. Every simply-plausible outcome is I-plausible. However, as  Example B illustrated, an outcome may be I-plausible even though it is not simply plausible. The following theorem characterizes the set of I-plausible outcomes.\footnote{Note that in case $x^*$ is simply plausible, \Cref{thm3} tells us that the set of 
 Cournot actions 
 in the lower contour set of $x^*$ is non-empty;
 let $x_n^C$ be such an action.  As $x_n^C=\phi(x_n^C)$ choosing  $x'=x''=x_n^C$  immediately provides two actions satisfying the conditions of   \Cref{thm:int_covers}.} All proofs for this subsection are in \Cref{A:subsecIplaus}.

\begin{theorem}\label{thm:int_covers} The set of  I-plausible actions is $U$-monotone.
An action $x^*$ is I-plausible if and only if the lower contour set of $x^*$ with respect to $U$ includes  actions $x'$ and $x''$ such that 
\begin{equation}\label{eqzaneetti1}
\phi(x')\leq x'\leq x''\leq \phi(x'').
\end{equation}
\end{theorem}

Suppose that  $x$ is I-plausible, and let $(K, \beta)$ be an admissible pair that implements $x$ in which the commitment structure  $K$ contains only intervals. Pick $x'$ in the upper contour set of $x$ with respect to $U$, and consider $K':=K \cup \big\{  \{x'\}\big\}$, and $\beta': K' \rightarrow \mathcal{X}$ such that $\beta'(\mathcal{X}_i)= \beta(\mathcal{X}_i)$ for all $\mathcal{X}_i \in K$ and $\beta'( \{x'\})=x'$. The commitment structure $K'$ contains only intervals. Furthermore, the pair $(K', \beta')$ is evidently  admissible and implements $x'$. We conclude that the I-plausible actions are $U$-monotone.

The \textit{if} part of the theorem is straightforward. Let $x^*$ be such that the lower contour set of $x^*$ with respect to $U$ includes  actions $x'$ and $x''$ satisfying \eqref{eqzaneetti1}. Consider $K=\big\{\{x^*\},[\underline{x},x''],[x',\overline{x}]\big\}$, and $\beta$ given by $\beta\big(\{x^*\}\big)=x^*$, $\beta\big([\underline{x},x'']\big)=x''$, and $\beta\big([x',\overline{x}]\big)=x'$. By \Cref{lem:extension}, the pair $(K,\beta)$ is admissible; it  implements $x^*$ since $U(x^*)\geq \max \{U(x'), U(x'')\}$.\footnote{In particular, it follows   that every I-plausible outcome is an equilibrium outcome of a game $G(K)$ for some commitment structure $K$ comprising at most three elements.}

The proof of the  \textit{only if} part of the theorem 
first establishes that if $x^*$ is I-plausible then some action  $x$ in the lower contour set of $x^*$ must satisfy  $\phi(x)\leq x$, and some action $x$ in the lower contour set of $x^*$ must satisfy  $\phi(x)\geq x$. Next, 
suppose that every action $x$ in the lower contour set of $x^*$ for which 
$\phi(x) \geq x$ is strictly smaller than every action $x$ in the lower contour set of $x^*$ for which $\phi(x) \leq x$, and pick $x^\dagger$ in between these two subsets of actions. Now let $(K, \beta)$  be an admissible pair such that $K$ is a CST comprising only intervals. By \Cref{lem:extension}, if $\mathcal{X}_i \in K$ is an interval containing $x^\dagger$ then either $\beta(\mathcal{X}_i) \geq x^\dagger$  and $\phi\big( \beta(\mathcal{X}_i) \big)\geq \beta(\mathcal{X}_i)$, or $\beta(\mathcal{X}_i) \leq x^\dagger$  and $\phi\big( \beta(\mathcal{X}_i) \big)\leq \beta(\mathcal{X}_i)$. Yet every action $x$ in the lower contour set of $x^*$ for which $\phi(x)\geq  x$ is strictly smaller than $x^\dagger$, while every action $x$ in the lower contour set of $x^*$ for which $\phi(x)\leq  x$ is strictly greater than $x^\dagger$. We conclude that $\beta(\mathcal{X}_i)$ must belong to the strict upper contour set of $x^*$. So $(K, \beta)$ cannot implement  $x^*$.

Note that 
if a Cournot action $x_{n}^C$ belongs to the lower contour set of an  action $x^*$, then setting $x'=x''= x_{n}^C$ in \Cref{thm:int_covers} proves  that $x^*$ is I-plausible. We thus obtain the following corollary:

\begin{corollary} \label{corocoussinmoel}
All actions in the  upper contour set of a Cournot action with respect to $U$ are I-plausible.
\end{corollary}

Applying \Cref{thm:int_covers} to the example of \Cref{F:fig1} shows that the I-plausible actions are $ \{0\} \cup \big[ 5/17,5/2 \big]$. In particular,  actions in the interval  $\big(5/9,5/4\big)$ are I-plausible but are not simply plausible.\footnote{An action $x^* \in \big(5/9,5/4\big)$ is for instance implemented by the pair $(K,\beta)$ where 
$K=\big\{\gX,[0,x^*]\big\}$, $\beta\big(\gX\big)=0$, 
and $\beta\big([0,x^*]\big)=x^*$. To see that no $x^* \in \big(0,5/17\big)$ is I-plausible, notice that the intersection between $\{x:\phi(x)\geq x\}$ and the lower contour set of  $x^*$ with respect to $U$  is empty.}

The previous analysis has shown that an outcome may be I-plausible even though it is not simply plausible and that I-plausible outcomes can give the leader a smaller payoff than her lowest Cournot payoff. However, in certain prominent cases, every I-plausible outcome does guarantee the leader at least her minimum Cournot payoff. For example, consider  a setting with a unique Cournot action $x^C$. We have in this case $\big\{x  : \phi(x) \geq x \big\}= [\underline{x}, x^C]$ and $\big\{x  : \phi(x) \leq x \big\}= [x^C, \overline{x}]$. Applying \Cref{thm:int_covers} thus shows that all I-plausible actions must belong to the upper contour set of $x^C$ with respect to $U$. As all actions in the upper contour set of $x^C$ are simply plausible (\Cref{cor:UCS_LCS}), we conclude that in this case an action is I-plausible if and only if it is in the upper contour set of $x^C$ with respect to $U$. A similar result holds when $U$ is either quasi-convex or quasi-concave.

\begin{proposition}\label{cor:2} 
If there exists a unique Cournot outcome, or if  $U$ is either quasi-convex or quasi-concave, then an action is I-plausible if and only if it belongs to the union of the upper contour sets of the Cournot actions with respect to $U$.
\end{proposition}

\subsection{Partitional Commitment Structures} \label{subsecPplaus}

We say that an outcome $(x,y)$ is \textit{P-plausible} if it is 
an equilibrium outcome of $G(K)$ for some commitment structure $K$ which partitions the leader's action space. Every simply-plausible outcome is  P-plausible. However, as Example A illustrated, an outcome may be P-plausible even though it is not simply plausible. We start this subsection by establishing that any plausible outcome is, in fact, P-plausible.\footnote{We thank an anonymous referee for pointing this out to us.} All proofs for this section are in \Cref{A:subsecPplaus}.

\begin{proposition} \label{P:ref1conjecture}
An outcome is plausible if and only if it is P-plausible. Moreover, the set of P-plausible outcomes is $U$-monotone.
\end{proposition}

 The \textit{if} part of the proposition is trivial. The gist of the proof of the \textit{only if} part is as follows. Let $(x,y)$ be plausible, and $(K, \beta)$ be an admissible pair that implements $(x,y)$. Suppose that $K=\{\mathcal{X}_1, \cdots, \mathcal{X}_n \}$, and\footnote{See the \Cref{A:subsecPplaus} for the general case.}  
 \begin{equation}\label{eqrefproof}
 \beta(\mathcal{X}_i) \notin \bigcup_{j\neq i} \mathcal{X}_j, ~~\text{for  $i=1, \cdots, n$.}
  \end{equation}
Now let $\mathcal{X}'_i=\mathcal{X}_i \setminus  \bigcup_{j<i} \mathcal{X}_j$, for $i=1, \cdots, n$, and $K'=\{\mathcal{X}'_i, \cdots, \mathcal{X}'_n\}$. Then $K'$ partitions $\mathcal{X}$ and \eqref{eqrefproof} implies $\beta(\mathcal{X}_i) \in \mathcal{X}'_i$, for  $i=1, \cdots, n$. Letting $\beta'(\mathcal{X}'_i)=\beta(\mathcal{X}_i)$ for  $i=1, \cdots, n$, we see that $(K', \beta')$ is admissible, since $\mathcal{X}'_i \subseteq \mathcal{X}_i$. Finally, $(K', \beta')$ evidently implements $(x,y)$,  so $(x,y)$ is P-plausible.

Since any plausible outcome is also P-plausible, 
the set of P-plausible outcomes is typically very large. To make progress, we restrict attention in the rest of this subsection to settings that satisfy the following three regularity conditions:
\begin{enumerate} \setlength{\itemindent}{1em}
\item [\textbf{(RC1)}] there exists a  unique and interior Cournot outcome;
\item [\textbf{(RC2)}] all externalities are either strictly positive or strictly negative;\footnote{Formally,  externalities are strictly positive (respectively, negative) if 
$u$ and $v$ are strictly increasing (respectively, decreasing) in their second arguments.}
\item [\textbf{(RC3)}] payoffs are either strictly supermodular or strictly submodular.\footnote{Formally, $u$ (and similarly $v$) is strictly supermodular (respectively, submodular) if 
$u(x',y')+u(x,y) > u(x',y)+u(x,y')$  (respectively,  $u(x',y')+u(x,y) < u(x',y)+u(x,y')$)  for all $x'>x$ and $y'>y$.
Supermodular payoffs capture strategic complementarities; submodular payoffs capture strategic substitutabilities.
}
\end{enumerate}
For instance, the  duopoly example of \Cref{subsec:duoexample} satisfies all three conditions as long as the  returns to scale are not too large ($r<1$). When $u$ and $v$ are twice differentiable, conditions (RC2) and (RC3) become  $~u_{2}v_{2}>0$ and   $~u_{12}v_{12}>0$, respectively. Slightly abusing notation, whether or not the payoffs are differentiable, $u_2>0$ will indicate positive  externalities, and $u_2<0$ negative externalities. Similarly, $u_{12}>0$ will indicate supermodular payoffs, and $u_{12}<0$ submodular payoffs.

As $u(\cdot, R_F(x))$ is strictly  quasi-concave, for  $x\leq x^C$ the set of actions which the leader would weakly prefer to $x$ when the follower best-responds to $x$ can be written as  $[x, \gamma(x)]$, for some  $\gamma(x)$.\footnote{If $x=x^C$ then $\gamma(x)=x^C$. Otherwise, either $u(\overline{x},R_F(x))> u(x,R_F(x))$ in which case $\gamma(x)=\overline{x}$, or $u(\overline{x},R_F(x))\leq  u(x,R_F(x))$ in which case $\gamma(x)$ is the unique action different from $x$ making the leader  indifferent between choosing $x$ and  $\gamma(x)$ when the follower best-responds to $x$.}
Similarly, for  $x\geq x^C$ 
 said set of actions can be written as  $[ \gamma(x),x]$.\footnote{If $x=x^C$ then $\gamma(x)=x^C$.
Otherwise,  either $u(\underline{x},R_F(x))> u(x,R_F(x))$ in which case $\gamma(x)=\underline{x}$, or $u(\underline{x},R_F(x))\leq  u(x,R_F(x))$ in which case $\gamma(x)$ is the unique action different from $x$ making the leader  indifferent between choosing $x$ and  $\gamma(x)$ when the follower best-responds to $x$.}
This defines a 
mapping $\gamma: \mathcal{X} \rightarrow \mathcal{X}$;
note that, as $u$ and $R_F$ are continuous, so is $\gamma$.
Finally,  let
\begin{equation*}
\mathcal{S}:=\begin{cases}
 \big\{ x:\;  x \leq  \gamma(x) \leq x^C  \big\}   & \text{if  $ u_2u_{12}>0$},\\
 \big\{ x:\;   x^C \leq  \gamma(x) \leq x   \big\} & \text{if  $ u_2u_{12}<0$}.
\end{cases}   
\end{equation*}
Since  $\gamma$ is continuous, the set $\mathcal{S}$ is compact.
Moreover, this set evidently contains $x^C$. We can now  characterize the set of plausible outcomes.

\begin{theorem}\label{proplast} 
Suppose (RC1)--(RC3) hold. The set of  plausible actions  coincides with the upper level set of $\underline{U}:= \min_{ x \in \mathcal{S}}\; U\big( \gamma(x) \big)$ with respect to $U$.\footnote{The upper level set of $\underline{U}$ with respect to $U$ is defined as $\{x: U(x) \geq \underline{U}\}$.}
\end{theorem}

 Our proof that 
 any   $x$ such that $U(x)\geq \underline{U}$
 is a 
 plausible action   rests on the following methodology.\footnote{The proof that the upper level set of $\underline{U}$  includes the set of  plausible actions   is  more technical. We refer the interested reader  to   the appendix.} 
 We first  collect all $x'$ in the strict upper contour set of $x$ with respect to $U$, and place these $x'$ into a set together with some worse alternative, say $x''$. Then, provided $(x'', R_F(x''))$ is a Nash equilibrium of the corresponding subgame, the action $x$ we started from can readily be made into an equilibrium outcome.

 To be more specific, 
consider the case in which $u_2>0$ and $u_{12}>0$, that is, players' actions have positive externalities and are strategic complements. In this case, 
$\mathcal{S}=\big\{ x:\;  x \leq  \gamma(x) \leq x^C  \big\} $, and 
the function $U$ is increasing in the range $[\underline{x}, x^C]$.\footnote{Intuitively, increasing the action of the leader induces the follower to increase her action too (due to strategic complementarities), and this benefits the leader (since externalities are positive).} Now
pick some action 
$\hat{x} \in \mathcal{S}$ and let $x^*$ be in the upper contour set of $\gamma(\hat{x})$ with respect to $U$;
we will argue that  $x^*$ is plausible. The case $\hat{x}=x^C$ being trivial, suppose $\hat{x}\neq x^C$ and so  $\hat{x} <  \gamma( \hat{x}) < x^C$.
Let $\mathcal{X}_1=[\underline{x},\hat{x}]\cup(\gamma(\hat{x}),x^*)\cup(x^*,\overline{x}]$.   
Next, let $K$ be the partition of the leader's action space containing $\mathcal{X}_1$, the interval $(\hat{x},\gamma(\hat{x})]$ and the singleton $\{x^*\}$. Lastly, let $\beta(\mathcal{X}_1)=\hat{x}$, $\beta((\hat{x},\gamma(\hat{x})])=\gamma(\hat{x})$ and $\beta(\{x^*\})=x^*$. We claim that the pair $(K, \beta)$ is admissible.
By definition of 
$\gamma$, 
the set of actions which the leader  weakly prefers  to $\hat{x}$ when the follower best-responds to $\hat{x}$ can be written as  $[\hat{x}, \gamma(\hat{x})]$.
So
$\eta( x , \hat{x})\leq 0$ for all $x\in \mathcal{X}_1$. Furthermore as $\gamma(\hat{x})<x^C$ then $\phi(\gamma(\hat{x}))>0$ and therefore $\eta(x,\gamma(\hat{x}))\leq 0$ for all $x\in(\hat{x},\gamma(\hat{x})]$. Finally, we claim that  $(K, \beta)$ implements $x^*$.
Indeed,   $
U(\hat{x})<
 U\big(\gamma(\hat{x}) \big)\leq U(x^*)$; the first inequality follows from 
$\hat{x}< \gamma(\hat{x})$ and the fact that  $U$ is increasing 
over $[\underline{x}, \gamma(\hat{x})]$; the second inequality is immediate,
 since $x^*$ is in the upper contour set of   $\gamma(\hat{x})$ by definition.\footnote{This construction immediately implies that all plausible outcomes are an equilibrium outcome of $G(K)$ for some  commitment structure $K$ with at most three elements.}

\Cref{F:figS} illustrates \Cref{proplast} in the context of the duopoly example from  \Cref{subsec:duoexample} with  $\costpar=4/5$. In panel A, the black curve represents the graph of the function $\phi$, which crosses the 45-degree line at $x^C=5/11$. In this example $u_2<0$ and $u_{12}<0$, so  $\mathcal{S}= \big\{ x:\; x \leq \gamma(x) \leq x^C \big\}$. The gray curve represents the graph of the function $\gamma$: we see that $\mathcal{S}=[0, x^C]$ and $\gamma(\mathcal{S})=\big[5/18, x^C\big]$. Panel B depicts the graph of the function $U$. Minimizing $U\big(\gamma(x)\big)$ over $\mathcal{S}$ shows that $\underline{U}=U\big(\gamma(0)\big)=U(5/18)$.
So the set of plausible actions is  $\big[5/18,x^\dagger\big]$, where $U(x^\dagger)=U(5/18)$. 
Note that $\underline{U}$ is less than the leader's Cournot payoff $U(x^C)$, hence some plausible outcomes give the leader less than her Cournot payoff.

\begin{figure}[!htb]
  \centering
  \subfloat[]{\includegraphics[width=0.47\textwidth,trim={16pt 11pt 21pt 8pt},clip]{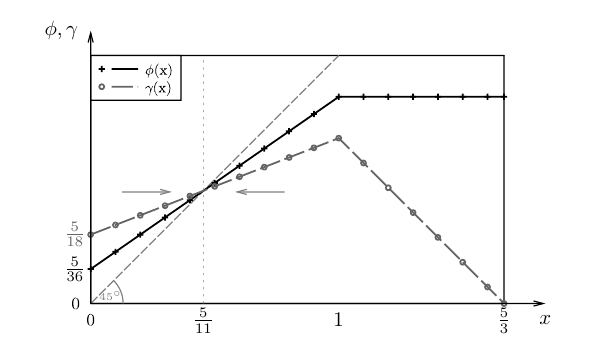}\label{F:phiS}}
  \hfill
  \subfloat[  ]{\includegraphics[width=0.47\textwidth,trim={16pt 11pt 21pt 8pt},clip]{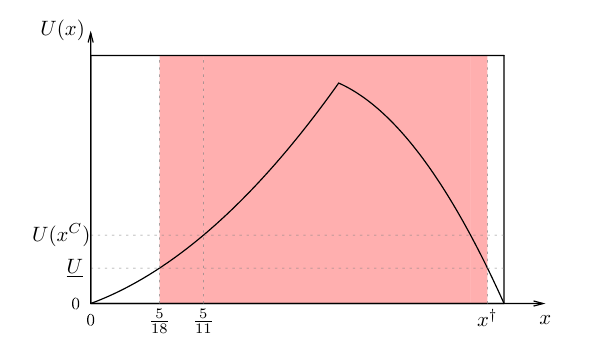}\label{F:US}}
  \caption{}
  \label{F:figS}
\end{figure}

 The question remains as to whether we can find conditions that guarantee   $\underline{U}<U(x^C)$. We show in \Cref{A:subsecPplaus} that when the payoff functions are twice  differentiable, a simple sufficient condition is given by $\gamma'(x^C)>0$. Calculations relegated to \Cref{A:subsecPplaus} establish that $\gamma'(x^C)>0$ if and only if $R_{L}'(y^C)R_F'(x^C)>1/2$.
 In other words,   whenever the best-response functions are sufficiently steep, some plausible outcomes give the leader less than her Cournot payoff. 
\begin{proposition}\label{remark:linear}
Suppose $u$ and $v$ are twice  differentiable and (RC1)--(RC3) hold. If $R_{L}'(y^C)R_F'(x^C)>1/2$ then $\underline{U}<U(x^C)$.
\end{proposition}

We have seen that the set of plausible outcomes may be larger than the set of simply-plausible ones. \Cref{P:ref1conjecture} shows that all plausible outcomes can be generated by partitional CSTs. A natural question is whether an even more restrictive class of CSTs can generate the entire set of plausible outcomes. The answer is yes, and one such class consists of what we call \textit{quasi-simple} CSTs:  a commitment structure is quasi-simple if it partitions the leader's action space and contains at most one element that is not an interval.
We say that an outcome is \textit{quasi-simply plausible} if it is an equilibrium outcome of $G(K)$ for some quasi-simple commitment structure $K$.

\begin{proposition}  
\label{propeasyfordes}
Suppose (RC1)--(RC3) hold. Then every plausible outcome is quasi-simply plausible.
\end{proposition}
The proof follows directly from the construction in the proof of \Cref{proplast}. In fact, it can be shown that if $U$ is either quasi-concave or quasi-convex then all plausible outcomes can be implemented with a partition of the leader's action space such that each partition element is either an interval or a union of two intervals. Effectively, these commitment structures are such that the leader plainly commits to choosing an action: (a) inside or outside an interval, (b) below or above a cutoff.




\subsection{Equilibrium Refinements}\label{subsec:refinements}

We saw that the Stackelberg and Cournot payoffs provide the bounds of the payoffs attainable by the leader under any simple CST. A natural question is whether some equilibrium refinement ensures that the Stackelberg and Cournot payoffs provide the bounds of the payoffs attainable by the leader under arbitrary CSTs.

Forward induction type of arguments eliminate some, but not all, 
equilibria giving the leader less than her Cournot payoffs.\footnote{ See \cite{myerson_game_1997} for a discussion of the merits and flaws of forward induction.} For instance, 
consider  the  setting of Example B at the beginning of this section, but this time with  commitment structure 
\[
\left\{ \left[\frac{1}{9}, \frac{4}{9}\right), \left[0, \frac{1}{9}\right) \cup \left[\frac{4}{9},1\right) \right\}.
\]
\Cref{F:example4} depicts the graph of $U$.  The subgame induced by the leader's choice of $[1/9, 4/9 )$ has a unique equilibrium, in which the leader chooses the action $1/9$. The other subgame has an equilibrium in which the leader picks $4/9$. As $U(4/9) > U(1/9)$, 
an equilibrium exists in which the leader chooses $4/9$. Yet, $U(4/9)<U(x^C_n)$, so the leader obtains a payoff smaller than her Cournot payoff. Since the subgame off the equilibrium path possesses a unique equilibrium, forward induction type of arguments have no bite.

\begin{figure}
\centering
\includegraphics[trim={12pt 0pt 21pt 0pt},clip,width=0.55\linewidth]{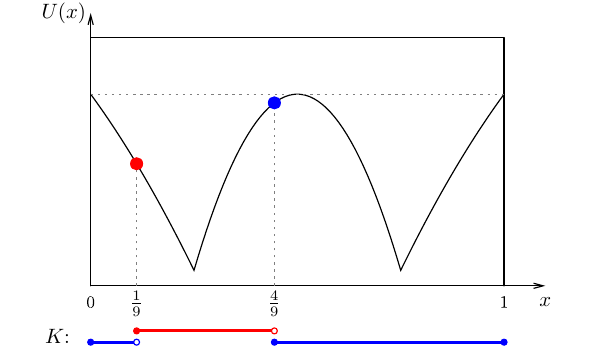}
\caption{} \label{F:example4}
\end{figure}

One alternative is to restrict attention to 
equilibria that select, in every period-2 subgame, the best continuation equilibrium from the perspective of the leader. In this case, any subgame induced by the leader's period-1 choice of a subset containing a Cournot action must give the leader a payoff at least as large as that Cournot payoff. Consequently, any such 
equilibrium ensures that the leader obtains at least her maximum Cournot payoff.

\section{Application} 
\label{sec:applications}

\begin{figure}[!h]
\centering
\includegraphics[width=0.6\textwidth,trim={34pt 11pt 5pt 7pt},clip]{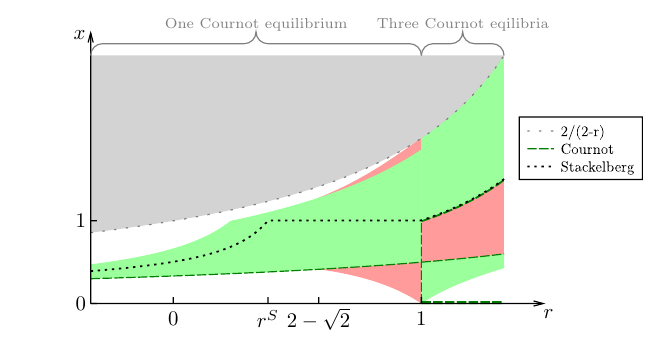}
\caption{} \label{AllOutcomes}
\end{figure}

We now apply our results to the duopoly setting introduced in \Cref{subsec:duoexample}. 
\Cref{AllOutcomes} shows the action set $\mathcal{X} = [0,2/(2-r)]$ as well as the Cournot and Stackelberg actions of the leader. For $r<1$, the unique Cournot outcome is such that leader and follower produce the same amount. For $r>1$ the three Cournot outcomes are, respectively, such that only the leader produces a positive amount, leader and follower produce the same amount, or only the follower produces a positive amount. Whenever $r<r^S$ 
the (unique) Stackelberg outcome is such that both firms produce a positive amount,  whereas 
if $r>r^S$ only the leader produces a positive amount.\footnote{For $r\in(r^S,1)$, the Stackelberg action is larger than the quantity a monopolist would choose;  the two are the same for   $r \geq 1$.}

The green area in \Cref{AllOutcomes} indicates   the set of simply plausible actions. For $r<1$, the green area corresponds to the upper contour set with respect to $U$ of the unique Cournot action (Corollary \ref{cor:UCS_LCS}); for $r>1$, the characterization of the green area is illustrated in the example of  \Cref{F:fig1}.  Every action $x$ in the green area is part of the equilibrium outcome for the simple commitment structure
\[
K=\biggl\{[0,x),\biggl[x,\frac{2}{2-r} \biggr] \biggr\}.
\]

The red area in \Cref{AllOutcomes} indicates the set of actions that are plausible but not simply plausible. For $r<1$, the characterization of this area relies on \Cref{proplast} and the observation that $\gamma(0)<x^C\Leftrightarrow r> 2-\sqrt{2}$;
for $r>1$, \Cref{thm:int_covers} implies that any action in between the  Stackelberg action and the second-highest Cournot action  is plausible.\footnote{Any action $x$ in the white area is not plausible as $u(x,R_F(x))<0$, and the leader can ensure zero profit with action $x=0$.}
Any action in the red area is part of an equilibrium outcome for a  quasi-simple CST.\footnote{For $r<1$, any action $x$ in the lower red area ($x<x^C$) is part of an equilibrium outcome for the commitment structure 
$K=\{(0,x],\{0\}\cup (x,2/(2-r)]\}$, while 
any $x$ in the upper red area ($x>x^C$)   for  $K=\{(0,\gamma(0)],\{0\}\cup(\gamma(0),x),[x,2/(2-r)]\}$. 
For $r>1$, any $x$ in the  red area is part of an equilibrium outcome for $K=
\{\{x\}, [0,x) \cup (x,2/(2-r) ]\}$. 
}

Next, we consider outcomes that maximize some familiar objective functions within the set of plausible outcomes. Formally, for a given objective function $W$, we solve the maximization problem:
\begin{equation}\label{eqproblemwelfarebisbis}
\tag{DP}
\text{    $\max  W\big(x,R_F(x)\big)~~$ s.t. $x$ is  plausible.}
\end{equation}
We first examine the  plausible outcomes associated with the largest attainable profit for one of the players (i.e., where $W=u$ or $W=v$). The Stackelberg outcome is plainly the best plausible outcome for the leader. On the other hand, since $v_2$ is here negative, the optimal plausible outcome for the follower involves the leader producing as little as plausibly possible. The proposition which follows summarizes these observations. In the rest of this section, $\underline{\underline{x}}$ (respectively, $\overline{\overline{x}}$) indicates the smallest (respectively, largest) plausible action of the leader.

\begin{proposition} \label{propdesign1} Suppose $u$ is given by \eqref{eqduopolypayoffs} and $v(y,x)=u(y,x)$.
If $W=u$, the unique solution of \eqref{eqproblemwelfarebisbis} is $x^S$.
If $W=v$, the unique solution of \eqref{eqproblemwelfarebisbis} is $\underline{\underline{x}}$.
\end{proposition}

The Stackelberg CST is optimal for the leader. The Cournot CST is optimal for the follower if and only if $\costpar \notin \big(2-\sqrt{2},1\big)$.

Next, we examine outcomes that maximize either consumer surplus, producer surplus, or total welfare (i.e., the sum of producer and consumer surplus). Consumer surplus and producer surplus are defined as
\begin{equation*}
CS(x,y) = \frac{(x+y)^2}{2},\;\;\;
PS(x,y)=u(x,y)+v(y,x).
\end{equation*}

\begin{proposition} \label{prop:3surpluses} Suppose $u$ is given by \eqref{eqduopolypayoffs} and $v(y,x)=u(y,x)$.
\begin{enumerate}[(i)]
\item If $W=CS$, the unique solution of \eqref{eqproblemwelfarebisbis} is 
$\overline{\overline{x}}$.
\item If $W=PS$, the unique solution of \eqref{eqproblemwelfarebisbis} is $x^C$ if $\costpar < \cbar$, and $x^S$ if $ \cbar<\costpar<1$;\footnote{$\cbar:=2-\frac{\sqrt[3]{ 3(9-\sqrt{78}) }}{3}-\frac{1}{\sqrt[3]{ 3(9-\sqrt{78}) }}$.} if $\costpar>1$ then the solutions are $x_3^C$ and $0$.\footnote{$x_3^C = \frac{1}{2-\costpar}$ is the highest of the three Cournot equilibria.}
\item If $W=CS+PS$, the unique solution of \eqref{eqproblemwelfarebisbis} is 
$\overline{\overline{x}}$.
\end{enumerate}
\end{proposition}

Part (i) of \Cref{prop:3surpluses} is explained as follows.\footnote{The proof of \Cref{prop:3surpluses}  is available upon request.} Firstly, we show that consumer surplus is a convex function of the quantity which the leader produces. The problem therefore reduces to choosing between $\underline{\underline{x}}$ and $\overline{\overline{x}}$. Inducing the leader to produce $\overline{\overline{x}}$ instead of $\underline{\underline{x}}$ is optimal as it exploits the strategic motive to produce large quantities. 

Part (ii) of \Cref{prop:3surpluses} is straightforward. With decreasing returns to scale, producer surplus is maximized by inducing both firms to produce the same quantity; in this case, the Cournot CST is producer-optimal. By contrast, with large returns to scale, producer surplus is maximized by letting one firm acquire a bigger market share than the other. In particular, for very large returns to scale, producer surplus is maximized when one firm acts as a monopolist. Consequently, the Cournot CST is producer-optimal for extreme returns to scale, whereas the Stackelberg CST is producer-optimal for sufficiently large returns to scale.

Part (iii) of \Cref{prop:3surpluses} follows from the fact that producer surplus tends to be less sensitive than consumer surplus to the quantity that the leader produces. Thus, maximizing total welfare implies maximizing consumer surplus.

\section{Conclusion}\label{sec: conclusion}

The Stackelberg leadership model assumes that the leader can commit to any action she might choose. Our paper takes a different view: we only assume that the leader can commit not to take certain subsets of actions.

We provide a tractable model of commitment that encompasses the Stackelberg and Cournot models as special cases but also enables us to capture situations of limited commitment. We characterize the set of outcomes resulting from all possible commitment structures, and shed light thereby on the ``limits of limited commitment". Our results highlight that, more than commitment, what matters is the precise \emph{form} that commitment takes. For instance, we show that whereas the Stackelberg and Cournot payoffs provide the bounds of the payoffs attainable by the leader under some appropriately defined class of ``simple" commitment structures, this property fails to hold more generally.

\bibliographystyle{apecon}
\bibliography{toomash-actiondisclosures}

\newpage

\appendix

\section{Appendix of Section \ref{sec: mainone}  
   } \label{A:intervals}

\begin{proof}[Proof of Lemma \ref{lem:extension}] 
We prove the \textit{only if} part of the lemma; the proof of the other part is similar. 
Suppose  that
 $(K,\beta)$ constitutes an admissible pair.
 Reason by contradiction, and suppose that we can find  $\mathcal{X}_i\in K$ such that 
 $ \phi\big(\beta(\mathcal{X}_i)\big)< \beta(\mathcal{X}_i)$ while 
  $\beta(\mathcal{X}_i) \neq \min \mathcal{X}_i$.
 The function 
 $\eta\big(\cdot,  \beta(\mathcal{X}_i)\big)$ is strictly quasi-concave, maximized at  $\phi\big(\beta(\mathcal{X}_i)\big)$, and satisfies 
 $\eta\big( \beta(\mathcal{X}_i),  \beta(\mathcal{X}_i)\big)=0$.
So $\eta\big(x,  \beta(\mathcal{X}_i)\big)>0$ for all $x \in \big[  \phi\big(\beta(\mathcal{X}_i) \big), \beta(\mathcal{X}_i) \big)$. 
Since 
$\mathcal{X}_i$ is an interval, 
  $\beta(\mathcal{X}_i) \in \mathcal{X}_i$, and 
  $\beta(\mathcal{X}_i) \neq \min \mathcal{X}_i$, we can find $\varepsilon>0$ such that 
  $\big(  \beta(\mathcal{X}_i)-\varepsilon, \beta(\mathcal{X}_i) \big) \subset \mathcal{X}_i$. Coupling the previous remarks shows the existence of $x \in  \mathcal{X}_i$ such that 
 $\eta\big(x,  \beta(\mathcal{X}_i)\big)>0$; this contradicts the assumption that 
 $(K,\beta)$ is admissible. Hence, 
 $ \phi\big(\beta(\mathcal{X}_i)\big)< \beta(\mathcal{X}_i)$ implies $\beta(\mathcal{X}_i)=\min \mathcal{X}_i$.
 Analogous arguments show that 
 $\phi\big(\beta(\mathcal{X}_i)\big)> \beta(\mathcal{X}_i)$ implies $\beta(\mathcal{X}_i)=\max \mathcal{X}_i$.
 \end{proof}

\begin{proof} [Proof of Theorem \ref{thm3}]
The \textit{if} part of the theorem was proven in the text; we prove here the  converse.
In the rest of the appendix, 
the upper contour set of $x$ with respect to $U$ will be denoted by $\mathcal{Q}_{\geq}(x)$, that is, 
\[
\mathcal{Q}_{\geq}(x):= \big\{ \tilde{x}: ~ U(\tilde{x}) \geq U(x) \big\}. 
\]
The sets $\mathcal{Q}_{<}(x)$, $\mathcal{Q}_{ \leq }(x)$, and $\mathcal{Q}_{>}(x)$ are similarly defined.

Pick an arbitrary 
  simply-plausible action $x^*$.
 We aim to  prove the existence of 
  a Cournot action $x_{n^*}^C\in \mathcal{Q}_{\leq} (x^*)$  such that
\eqref{eq:thm1} holds. If  $\phi(x^*)=x^*$, just take $x_{n^*}^C=x^*$; we treat below   the case in which
$\phi(x^*)>x^*$ (the remaining  case is analogous). Reason by contradiction, and suppose that
\begin{equation}\label{eqemrlondon}
    \mathcal{X}^C \cap (x^*, \overline{x}] \cap   \mathcal{Q}_{\leq} (x^*) =\varnothing.
\end{equation}
Let   $(K,\beta)$ be an admissible pair that implements 
 $x^*$, with $K$  a simple CST.   We will  show that 
 $K$ cannot be    finite.
   By Berge's maximum theorem,
both $R_F$ and $R_L$  are continuous, thus
 $\phi$ is continuous as well.
  As $\phi(x^*)>x^*$ and $\phi(\overline{x}) \leq \overline{x}$, the intermediate value theorem shows that 
    \[
    \mathcal{X}^C  \cap (x^*, \overline{x}] \neq \varnothing.
    \] 
    Note that
the continuity of the function $\phi$ implies    the compactness of  $\mathcal{X}^C$. So   $\mathcal{X}^C  \cap (x^*, \overline{x}] = \mathcal{X}^C  \cap [x^*, \overline{x}]$ possesses  a smallest element, that we denote by 
 $x_1^C$. Let $\mathcal{X}_1$ be the member  of $K$  containing $x_1^C$.
Then  Lemma 
\ref{lem:extension}  combined with \eqref{eqemrlondon}  gives 
\[
\beta(\mathcal{X}_1) \in (x_1^C, \overline{x}] \cap \big\{x: \phi(x)>x\big\}.
\]
Now let  $x_2^C$ be the smallest Cournot action greater than $
\beta(\mathcal{X}_1)$, and  denote by $\mathcal{X}_2$ the member  of $K$  containing $x_2^C$. The same logic as above gives  
$\beta(\mathcal{X}_2) \in (x_2^C, \overline{x}] \cap \big\{x: \phi(x)>x\big\}$, and so on.
If $K $ were finite, the previous iteration would have to  end after finitely many steps, say $m$. But then $\beta(\mathcal{X}_m)=\overline{x}$ and 
  $\beta(\mathcal{X}_m) \in  \big\{x: \phi(x)>x\big\}$, 
 giving $\phi(\overline{x})>\overline{x}$. The previous contradiction
  proves that 
$K$ cannot be  finite.

We proceed to  show that $K$ cannot be  infinite either.
The function $U$ is continuous and, by  \eqref{eqemrlondon}, 
$U(x_n^C)>U(x^*)$ for all $x_n^C \in \mathcal{X}^C \cap (x^*, \overline{x}]$. 
Furthermore, as already pointed out above, 
$\mathcal{X}^C \cap (x^*, \overline{x}]$ is  a compact set. 
Therefore, 
\begin{equation} \label{eqharrogate1}
\Delta:= \min_{x_n^C \in\; \mathcal{X}^C \cap (x^*, \overline{x}]} U(x_n^C) -U(x^*)>0.
\end{equation}

Next, $U$ being continuous and $\mathcal{X}$ compact, the function $U$ is uniformly continuous on $\mathcal{X}$. 
We can thus find $\delta>0$
such that 
$|U(x')-U(x)|<\Delta$ whenever  $|x'-x|<\delta$.
By \eqref{eqharrogate1}, we thus have
\begin{equation}\label{eqharrogate2}
U(x)>U(x^*),~ \text{for all $x$ such that  $|x-x_n^C|<\delta$, $x_n^C \in\; \mathcal{X}^C \cap (x^*, \overline{x}]$.}
\end{equation}

Now, since  $(K, \beta)$ implements $x^*$,  we must have  $U\big(\beta(\mathcal{X}_i)\big) \leq U(x^*)$ for all $\mathcal{X}_i \in K$.   
So \eqref{eqharrogate2} shows that each member of the sequence $\mathcal{X}_1, \mathcal{X}_2, \dots$ defined in the first part of the proof 
must have a length $\delta$ or more. This in turn implies that 
said sequence can have no more than $\frac{\overline{x}-x^*}{\delta}$ terms.
Yet we showed previously that this sequence cannot be finite. This contradiction completes the proof of the theorem.
\end{proof}

\section{Appendix of Subsection \ref{subsecIplaus}  } \label{A:subsecIplaus}

\begin{proof} [Proof of Theorem \ref{thm:int_covers}] 
 We prove here the  \textit{only if} part of the theorem; the rest was proven in the text.
Pick an arbitrary action $x^*$ of the leader.
Suppose that  $\mathcal{Q}_{\leq}(x^*) \cap \big\{x  : \phi(x) \leq x \big\} = \varnothing$.
Applying  Lemma 
\ref{lem:extension} shows that any admissible pair
$(K,\beta)$ in which the CST  $K$ contains only intervals
 must be such that 
$\beta(\mathcal{X}_i) \in   \big\{x  : \phi(x) \leq  x \big\}$ for every $\mathcal{X}_i \in K$ containing $\overline{x}$.
This, in turn, implies that every  
I-plausible action belongs to 
  $\mathcal{Q}_{>}(x^*)$, whence $x^*$ cannot be  I-plausible. A similar argument shows that  $\mathcal{Q}_{\leq}(x^*) \cap \big\{x  : \phi(x) \geq x \big\} = \varnothing$ implies that $x^*$ is not  I-plausible. 
Next, suppose that 
$\mathcal{Q}_{\leq}(x^*) \cap \big\{x  : \phi(x) \leq x \big\}$ and  $\mathcal{Q}_{\leq}(x^*) \cap  \big\{x : \phi(x) \geq x \big\}$ are  non-empty. 
Both $\phi$ and $U$ being continuous, the min and max of \eqref{eqzaneetti1} are in this case well defined (since $\mathcal{X}$ is a compact set). Suppose that  $\max  \mathcal{Q}_{\leq}(x^*) \cap \big\{x  : \phi(x) \geq x \big\}<  \min \mathcal{Q}_{\leq}(x^*) \cap \big\{x  : \phi(x) \leq x \big\}$, 
 and
pick
\begin{equation}\label{eqzaneetti2}
x^\dagger   \in \big(\max  \mathcal{Q}_{\leq}(x^*) \cap \big\{x  : \phi(x) \geq x \big\},  \min \mathcal{Q}_{\leq}(x^*) \cap \big\{x  : \phi(x) \leq x \big\}  \big).
\end{equation} 
Applying  Lemma 
\ref{lem:extension} shows that any admissible pair
$(K,\beta)$ comprising an interval CST must be such that, 
 for every $\mathcal{X}_i \in K$ containing $x^\dagger$, either (i)
$\beta(\mathcal{X}_i) \in   \big\{x\geq x^\dagger  : \phi(x) \geq  x \big\}$ or (ii)
$\beta(\mathcal{X}_i) \in   \big\{x\leq x^\dagger  : \phi(x) \leq  x \big\}$. So  \eqref{eqzaneetti2} gives 
$\beta(\mathcal{X}_i) \in  \mathcal{Q}_{>}(x^*)$. It ensues that  $x^*$ cannot be  I-plausible.
\end{proof}

\begin{proof}[Proof of Proposition \ref{cor:2}] 
By \Cref{corocoussinmoel}, an action that belongs to the upper contour set of some Cournot action is I-plausible. Below we show that the converse is true too if $U$ is either quasi-convex or quasi-concave. 

Suppose that  $U$ is quasi-convex, and consider an action $x^*$  in the strict lower contour set of every Cournot action. Then $\mathcal{Q}_{\leq}(x^*)$ is a convex set, and $\phi(x)\neq x$ for all $x \in \mathcal{Q}_{\leq}(x^*)$. The intermediate value theorem shows that either $x<\phi(x)$ for all $x\in \mathcal{Q}_{\leq}(x^*)$, or  $x>\phi(x)$ for all $x\in \mathcal{Q}_{\leq}(x^*)$. Either way, \Cref{thm:int_covers} shows that  $x^*$ cannot be I-plausible.
 
Next, suppose that $U$ is quasi-concave, and consider an action $x^*$  in the strict lower contour set of every Cournot action. Then $\mathcal{Q}_{>}(x^*)$ is a convex set, and $\phi(x)\neq x$ for all $x \in \mathcal{Q}_{\leq}(x^*)$. This implies that, given $x\in \mathcal{Q}_{\leq}(x^*)$, either (i) $\phi(x)>x$ and  $x<x^C_n$ for all $x^C_n\in \mathcal{X}^C$, or  (ii) $\phi(x)<x$ and  $x>x^C_n$ for all $x^C_n\in \mathcal{X}^C$. We conclude using \Cref{thm:int_covers} that $x^*$ is not I-plausible. 
\end{proof}

\section{Appendix of Subsection \ref{subsecPplaus}  } \label{A:subsecPplaus}

\begin{proof}[Proof of \Cref{P:ref1conjecture}:] 
%
The \textit{if} part of the proposition is evident. We prove here the \textit{only if} part. 
Let $(x^*,y^*)$ be plausible, and  $(K, \beta)$ be  an admissible pair that implements $(x^*,y^*)$. Let $B$ denote the image of the function $\beta$, that is, $B:=\beta(K)$. For each action $x \in \mathcal{X}\setminus B $, let 
\[
B_x:=\bigcup_{\text{$\mathcal{X}_i \in K$ s.t. $ x \in \mathcal{X}_i$}} \{\beta(\mathcal{X}_i)\},
\]
that is, $B_x$ is the set of all leader's actions played in some subgame containing $x$. If $K$ is not a partition, $B_x$ may contain more than one element. For all  $x\in \mathcal{X}\setminus B $, let $\rho(x)$ be one element of $B_x$.%
\footnote{Using the axiom of choice, we define a mapping $\hat{\rho}$ that selects $\hat{\rho}(B_x) \in B_x$ for each $x \in \mathcal{X} \setminus B$. Then, we define $\rho : \mathcal{X} \setminus B \to B$ by setting $\rho(x) = \hat{\rho}(B_x)$.}
For all $b \in B$, 
we can now define the set
\[
[b]:= \{b\} \cup \bigl\{ x \in \gX \mid \rho(x) = b \bigr\}.\]
Finally,  let
\[
K':=\{[x] \mid x \in B\}, 
\]
and  $\beta': K' \rightarrow \mathcal{X}$ such that $\beta'([x])=x$, for all $x \in B$. By construction, $K'$ is a partition of $\gX$.

Since $(K, \beta)$ is admissible, we have 
$ \eta(\tilde{x},x)\leq 0$ for all $\tilde{x} \in B_x$,
which implies that the pair $(K', \beta')$ is admissible as well. Moreover, it implements $(x^*, y^*)$ since $\beta'(K')=B=\beta(K)$.
\end{proof}

\begin{lemma} \label{lemmapouvoirachat}
Suppose (RC1)--(RC3) hold. If $ u_2 u_{12}>0$, then $U$ is increasing over $[ \underline{x}, x^C]$. If $ u_2 u_{12}<0$, then $U$ is decreasing over $  [x^C, \overline{x}]$.
 \end{lemma}
 
\begin{proof}
We show the proof for the case in which $u_2>0$ and $u_{12}>0$; the other cases are similar.\footnote{Recall, $u_2 >0$ is shorthand  notation for positive externalities and $u_{12}>0$ for supermodular payoffs.
}
Pick an arbitrary   $x< x^C$, and $\varepsilon>0$ sufficiently small that $u\big(x+\varepsilon, R_F(x)\big)> u\big(x, R_F(x)\big)$.\footnote{The function 
$u\big(\cdot, R_F(x)\big)$ being strictly quasi-concave and maximized at $\phi(x)$,
it ensues that
$x<x^C$ implies  $\eta( x+\varepsilon,x)>0$ for all  sufficiently small $\varepsilon>0$.}
Then, $R_F$ being non-decreasing (since $v_{12}>0$)  and $u_2>0$:
\begin{equation*}\label{eqcafelatte1bis}
U(x+\varepsilon)=u\big(x+\varepsilon, R_F(x+\varepsilon) \big) \geq u\big(x+\varepsilon, R_F(x) \big)> u\big(x, R_F(x) \big)=U(x).
\end{equation*}
\end{proof}

\begin{lemma} \label{lemmaHMRevenue}
Suppose (RC1)--(RC3) hold. Then
\begin{equation} \label{eqcocodrink}
\mathcal{S}=  \big\{ x: \; \eta( x^C,  x) \leq 0 \big\}   \cap \big\{x : \;  u\big(x^C, R_F(x)\big) \leq  U(x^C) \big\}.
\end{equation}
\end{lemma}

\begin{proof}
We show the proof of the lemma for the case $u_2>0$ and $u_{12}>0$ (the other cases are similar).
Recall that in this case 
$\mathcal{S}:=
 \big\{ x:\;  x \leq  \gamma(x) \leq x^C  \big\}$.

The function   $R_F$ being in this case non-decreasing (and, indeed, increasing in a neighborhood of $x^C$ since $y^C \in \text{int}  (\mathcal{Y})$)
and $u_2>0$, notice that 
\[
 u\big(x^C, R_F(x)\big) >u \big(x^C, R_F(x^C)\big)=U(x^C), ~~\text{for all $x>x^C$}.
\]
So 
   $u\big(x^C, R_F(x)\big) \leq  U(x^C)$  implies $x\leq x^C$.
Now consider $x\leq x^C$ such that  $\eta( x^C,  x) \leq 0$.
We will show that $x \in \mathcal{S}$. If $x=x^C$ the previous claim  is immediate, so
pick
 $x< x^C$.
 The function  
 $\eta(\cdot, x)$ is 
 strictly quasi-concave, and maximized    at $\phi(x)>x$.\footnote{As
 $\phi$ is continuous, notice that
\begin{equation}\label{ineqrep}
\begin{cases}
  \phi(x)>x & \text{for $x<x^C$},\\
   \phi(x)<x & \text{for $x>x^C$}.
          \end{cases}     
\end{equation}
  }
 As $\eta(x,x)=0 \geq \eta(x^C, x)$, we see  by definition of $\gamma(x)$ that $x< \gamma(x)\leq x^C$.
 The right-hand side of \eqref{eqcocodrink} is thus contained in the set  $\mathcal{S}$. The proof of the 
 reverse inclusion is analogous.
\end{proof}

\begin{lemma}\label{lemmaonyxtabl}
Suppose (RC1)--(RC3) hold, and $\mathcal{S}=\{x^C\}$. Then all plausible actions  belong to  the 
upper contour set of $x^C$ with respect to $U$. 
\end{lemma}

\begin{proof} Reason by contradiction, and suppose that
some action 
$x^* \in  \mathcal{Q}_{<}(  x^C)$
is  plausible.
Let 
$(K, \beta)$ be an admissible pair that implements $x^*$. Choose   an element  $\mathcal{X}_i$ of the CST $K$ such that    $x^C \in  \mathcal{X}_i$. Using
 Lemma \ref{rem:remark1}  yields
$\beta(\mathcal{X}_i) \in ~    \big\{ x: \; \eta( x^C,  x) \leq 0 \big\} \cap \mathcal{Q}_{\leq}(  x^*)$, 
and, since $x^* \in  \mathcal{Q}_{<}(  x^C)$,
\begin{equation}\label{eqdevereorchard1}
\beta(\mathcal{X}_i) \in ~    \big\{ x: \; \eta( x^C,  x) \leq 0 \big\} \cap \mathcal{Q}_{<}(  x^C).
\end{equation}
 In turn, \eqref{eqdevereorchard1} 
 yields
\begin{equation}\label{eqdevereorchard2}
u\Big(x^C, R_F\big(\beta(\mathcal{X}_i)\big)\Big) \leq u\Big(\beta(\mathcal{X}_i),R_F\big(\beta(\mathcal{X}_i)\big) \Big)= U\big(\beta(\mathcal{X}_i)\big)<U(x^C).
\end{equation}
Coupling 
 \eqref{eqdevereorchard1} and \eqref{eqdevereorchard2} gives
\[
\beta(\mathcal{X}_i) \in   \big\{ x: \; \eta( x^C,  x) \leq 0 \big\}   \cap \big\{x : \;  u\big(x^C, R_F(x)\big)< U(x^C) \big\}.
\] 
Applying 
 \Cref{lemmaHMRevenue}, we obtain $
\beta(\mathcal{X}_i) \in \mathcal{S} \setminus \{x^C \}$, contradicting $\mathcal{S}=\{x^C\}$.
\end{proof}

\begin{lemma}\label{lemmaemercall1}
Suppose (RC1)--(RC3) hold. Assume $u_{12}>0$ and $u_2>0$. Consider 
an admissible pair 
 $(K, \beta)$ which   implements some action $x^*$. Then, if
 $x \in \mathcal{Q}_{>}(  x^*)$, we have
 $\beta(\mathcal{X}_i )<x$ for every  
$\mathcal{X}_i \in K$ which contains $x$. 
\end{lemma}

\begin{proof} Let  $x \in \mathcal{Q}_{>}(  x^*)$, and 
pick an arbitrary $\mathcal{X}_i \in K$ containing $x$.
Reason by contradiction, and suppose that $\beta(\mathcal{X}_i) \geq x$. Then,  $R_F$ being non-decreasing (since $v_{12}>0$) and $u_2>0$, we obtain 
  \begin{equation}\label{eqspringwater1}
 u\Big(x,R_F\big(\beta(\mathcal{X}_i)\big)\Big)\geq u\big(x,R_F( x)\big)>u\big(x^*,R_F(x^*)\big).
  \end{equation}
 Since  $(K, \beta)$  is admissible, we also have
  \begin{equation}\label{eqspringwater2}
  u\Big( \beta(\mathcal{X}_i), R_F\big(\beta(\mathcal{X}_i) \big)\Big) \geq  u\Big(x,R_F\big(\beta(\mathcal{X}_i)\big)\Big). 
  \end{equation}
  Coupling \eqref{eqspringwater1} and \eqref{eqspringwater2} yields  \begin{equation*}\label{eqspringwater3}
  u\Big( \beta(\mathcal{X}_i), R_F\big(\beta(\mathcal{X}_i) \big)\Big) >
 u\big(x^*,R_F(x^*)\big).
   \end{equation*}
The previous   inequality contradicts  the assumption that 
 $(K, \beta)$ implements $x^*$.
 \end{proof}

\begin{lemma}\label{lemmaemercall2} Suppose (RC1) holds.
Let  
 $(K, \beta)$ be an admissible pair.
If   $\beta(\mathcal{X}_i)< \min\{x^C,x\}$
 for some $\mathcal{X}_i \in K$ which contains $x$, 
    then $\gamma \big( \beta(\mathcal{X}_i) \big) \in \big( \beta(\mathcal{X}_i),  x \big]$.
\end{lemma}

\begin{proof} Pick 
$x \in \mathcal{X}$, and 
$\mathcal{X}_i \in K$ containing $x$.  
Since 
   $(K, \beta)$ is admissible:
  \begin{equation}\label{eqBTphone2}
  \eta\big( x, \beta(\mathcal{X}_i)\big) \leq 0.
    \end{equation}
Now 
  suppose that 
 $\beta(\mathcal{X}_i)< \min\{x^C,x\}$.
 In this case, 
 the strictly concave  function  
 $\eta\big(\cdot, \beta(\mathcal{X}_i) \big)$ 
 attains (by virtue of   \eqref{ineqrep}) a  maximum   at $\phi\big(\beta(\mathcal{X}_i)\big)> \beta(\mathcal{X}_i)$.  From  \eqref{eqBTphone2} and  the fact that   $\beta(\mathcal{X}_i)<x$ we  obtain (by definition of $\gamma$) 
  $\beta(\mathcal{X}_i)< \gamma\big(\beta(\mathcal{X}_i) \big)\leq x$.
\end{proof}

\begin{proof}[Proof of Theorem \ref{proplast}]
Start with the case  $\mathcal{S}=\{x^C\}$.
Combining    \Cref{cor:UCS_LCS}, \Cref{P:ref1conjecture}, and   
 \Cref{lemmaonyxtabl}  shows 
 that 
  the set of 
simply-plausible  actions,  the set of I-plausible actions, the set of P-plausible actions, and the set of plausible actions all coincide with 
 the upper contour set of $x^C$.

The remainder of the proof 
deals with the case 
$\mathcal{S} \supsetneq\{x^C\}$.
 Below,  assume $u_{12}>0$ and $u_2>0$ (the other cases are analogous).
Recall that in this case 
$\mathcal{S}:=
 \big\{ x:\;  x \leq  \gamma(x) \leq x^C  \big\}$.
 The function 
 $\gamma$ being  continuous,  $\mathcal{S}$ is  a compact set.
 By Lemma \ref{lemmapouvoirachat},
 we can  thus
find 
$\hat{x} \in 
\mathcal{S}$ with $\hat{x}<x^C$ and
\begin{equation}\label{eqcafedisgner}
 U\big(\gamma(\hat{x})\big) = \min_{ x \in \mathcal{S}} U\big( \gamma(x)\big).
\end{equation}
 To shorten notation, let $\hat{\gamma}:=\gamma(\hat{x})$; as $\hat{x}<x^C$, 
note that, by definition of $\gamma$,  
 \begin{equation}\label{eqmarkerred}
     \hat{x}< \hat{\gamma}\leq x^C.
 \end{equation}
 We proceed to show  that (a)
  all actions in 
$\mathcal{Q}_{\geq}(\hat{\gamma})$ are P-plausible, and (b)
any plausible action belongs to 
$\mathcal{Q}_{\geq}(\hat{\gamma})$.

\noindent 
\underline{\textit{All actions in 
$\mathcal{Q}_{\geq}(\hat{\gamma})$   are P-plausible.}} We know by 
 \Cref{cor:UCS_LCS}  that  all actions in 
$\mathcal{Q}_{\geq}(  x^C)$ are  simply plausible.
So pick an action
$x^* \in \mathcal{Q}_{\geq}(\hat{\gamma}) \setminus \mathcal{Q}_{\geq}(  x^C)$ (if there exists none, we are done).
Define  
\[
\mathcal{X}_1:=  \{\hat{x}\} \cup \mathcal{Q}_{>}(x^*),\]
and let  $K$ denote the partition of $\mathcal{X}$ made up of  $
\mathcal{X}_1$, and only singletons besides  $
\mathcal{X}_1$. Lastly, let $\beta: K\rightarrow \mathcal{X}$ be given by $\beta(\mathcal{X}_1)= \hat{x}$ and 
$\beta(\{x\})=x$ for all $ x \in \mathcal{X} \setminus \mathcal{X}_1$.
We now  show that $(K, \beta)$ constitutes an admissible pair; notice that this   amounts to showing that 
\begin{equation}\label{eqlagerloef}
    \eta(\tilde{x}, \hat{x})\leq 0, ~~~
 \text{ for all $  \;\tilde{x} \in  \mathcal{X}_1$}.
 \end{equation}
As $x^* \in \mathcal{Q}_{\geq}(\hat{\gamma})$, 
 any $\tilde{x} \in \mathcal{Q}_{>}(x^*)$  belongs to $\mathcal{Q}_{\geq}(\hat{\gamma})$. On the other hand, 
since
  $\hat{\gamma} \leq x^C$ (see \eqref{eqmarkerred}), 
  Lemma \ref{lemmapouvoirachat} shows that 
  every $\tilde{x} \in \mathcal{Q}_{>}(x^*)$ satisfies  $\tilde{x} \geq \hat{\gamma}$.  Now, 
   the function  
 $\eta(\cdot, \hat{x})$ is 
 strictly quasi-concave, with  $\eta(  \hat{x}, \hat{x})= \eta(\hat{\gamma}, \hat{x})=0$; it thus follows from \eqref{eqmarkerred} 
 that 
 $ \eta(\tilde{x}, \hat{x})\leq 0$ for all 
$\tilde{x} \geq \hat{\gamma}$. Combining the previous observations establishes
\eqref{eqlagerloef}; so  $(K, \beta)$  is admissible.

Finally, 
  coupling \eqref{eqmarkerred} and 
   Lemma \ref{lemmapouvoirachat} yields
$U( \hat{\gamma})>U( \hat{x})$, giving in turn  $U(x^*) >U(\hat{x})= U\big(  \beta(\mathcal{X}_1)\big)$ (since 
$x^* \in \mathcal{Q}_{\geq}(\hat{\gamma})$).
We conclude  that 
  $(K, \beta)$  implements $x^*$, since $\mathcal{X} \setminus \mathcal{X}_1 \subset \mathcal{Q}_{\leq }(x^*)$.

\noindent 
\underline{\textit{All plausible actions belong to 
$\mathcal{Q}_{\geq}(\hat{\gamma})$.}}
 Reason by contradiction, and suppose that some plausible action  $x^*$  belongs to  $\mathcal{Q}_{<}(\hat{\gamma})$.
 Combining   
 \eqref{eqmarkerred}, 
   Lemma \ref{lemmapouvoirachat}, and  the fact that  $U$ is continuous  shows 
that we can  find an action, say  $x^\dagger$,   such that:
 \begin{equation}\label{eqboutcocacola1}
     x^\dagger<\hat{\gamma}, 
      \end{equation}
      and 
\begin{equation}\label{eqboutcocacola2}
    x^\dagger \in \;\mathcal{Q}_{>}(  x^*) \cap \mathcal{Q}_{<}(  \hat{\gamma}).
    \end{equation}

\noindent Now  consider a pair 
  $(K, \beta)$ which implements $x^*$, and $\mathcal{X}_i$ an element of the CST $K$  containing $x^\dagger$. By virtue of \eqref{eqboutcocacola2}, applying
  Lemma \ref{lemmaemercall1} shows that 
   \begin{equation}\label{eqboutcocacola3}
\beta(\mathcal{X}_i)<x^\dagger.
\end{equation}
On the other hand,  \eqref{eqmarkerred} and  \eqref{eqboutcocacola1} 
 show that 
    \begin{equation*}
 x^\dagger<\hat{\gamma} \leq x^C.
\end{equation*} 
 Hence, 
  Lemma \ref{lemmaemercall2} gives
\begin{equation}\label{eqkinder0}
\beta(\mathcal{X}_i) < \gamma\big( \beta(\mathcal{X}_i)\big) \leq x^\dagger< \hat{\gamma} \leq x^C.
\end{equation}
We thus obtain,  firstly, 
\begin{equation}\label{eqkinder1}
\beta(\mathcal{X}_i) \in \mathcal{S},
\end{equation}
and, secondly (using Lemma 
\ref{lemmapouvoirachat}),
\begin{equation}\label{eqkinder2}
U\Big( \gamma\big( \beta(\mathcal{X}_i)\big) \Big)< U(\hat{\gamma}). 
\end{equation}
The combination of \eqref{eqkinder1} and \eqref{eqkinder2}
contradicts \eqref{eqcafedisgner}.  Therefore, every  plausible action must  belong to 
$\mathcal{Q}_{\geq}(\hat{\gamma})$.
\end{proof}

\begin{proof}[Proof of Proposition \ref{remark:linear}]
By definition of $\gamma$:
    $\eta\big( \gamma(x), x  \big)=0$ for all  $x$ in  
 some   neighborhood  $O$  of $x^C$. We thus have
 \begin{equation*}
     u\big(  \gamma(x), R_F(x)  \big)=u\big(  x, R_F(x) \big), ~~~ \forall x\in O.
 \end{equation*}
Differentiating the previous expression  with respect to $x$ yields
\[
u_1\big(   \gamma(x), R_F(x) \big) \gamma'(x) +  u_2\big(   \gamma(x), R_F(x) \big)  R'_F(x)=  u_1\big(   x, R_F(x) \big)  +  u_2\big(   x, R_F(x) \big)  R'_F(x), 
\]
and, therefore, 
\begin{equation}\label{eqinnocent1}
\gamma'(x)=  \frac{ u_1\big(   x, R_F(x) \big)  +R'_F(x)   \big[  u_2\big(   x, R_F(x) \big) -  u_2\big(   \gamma(x), R_F(x) \big)   \big] }{   u_1\big(   \gamma(x), R_F(x) \big)},  ~~~ \forall x\in O  \setminus \{x^C\}.
\end{equation}
The numerator and denominator on the right-hand side of \eqref{eqinnocent1} tend to $0$ as $x \rightarrow x^C$. Then, by virtue  of L'H\^opital's rule  and using the fact that $\gamma(x)\rightarrow x^C$ as $x \rightarrow x^C$:
\begin{equation}\label{eqinnocent2}
\lim\limits_{x \rightarrow x^C}  \gamma'(x)= \lim\limits_{x \rightarrow x^C}
\frac{  u_{11}\big(  x, R_F(x) \big)  +  2  u_{12}\big(  x, R_F(x) \big) R'_F(x)   -    u_{12}\big(  x, R_F(x) \big) R'_F(x) \gamma'(x) }{  u_{11}\big(  \gamma(x), R_F(x) \big) \gamma'(x)  +    u_{12}\big(   \gamma(x), R_F(x) \big)  R'_F(x)}.
\end{equation}
On the other hand, in a neighborhood of $y=y^C$:
\[
R'_L(y)=  \frac{-u_{12}\big( R_L(y),y \big)  }{u_{11}\big( R_L(y),y \big)}.
\]
Therefore, 
\begin{equation}\label{eqinnocent3}
R'_L(y^C)=  \frac{-u_{12}( x^C,y^C )  }{u_{11}( x^C,y^C)}   =   \lim\limits_{x \rightarrow x^C}   \frac{-  u_{12}\big(  x, R_F(x) \big)}{ u_{11}\big(  x, R_F(x) \big)} =   \lim\limits_{x \rightarrow x^C}   \frac{-  u_{12}\big( \gamma( x), R_F(x) \big)}{ u_{11}\big(  \gamma(x), R_F(x) \big)}.
\end{equation}
Combining \eqref{eqinnocent3} with \eqref{eqinnocent2} gives
\[  \gamma'(x^C)=    \frac{1-2  R'_L(y^C)  R'_F(x^C)  + R'_L(y^C)  R'_F(x^C) \gamma'(x^C)  }{\gamma'(x^C) - R'_L(y^C)  R'_F(x^C)}.
\]
So $\gamma'(x^C)$  is a solution of 
\[  Z(Z-2 \alpha)=1-2 \alpha, 
\]
where $\alpha:= R'_L(y^C)  R'_F(x^C)$. So
either 
$\gamma'(x^C)=1$
or
$\gamma'(x^C)=2 \alpha -1$, whence   $\gamma'(x^C)>0$ if   $R_{L}'(y^C)R_F'(x^C)>1/2$.

Now suppose that $u_{12}u_2>0$ (the other case is similar), so that 
$\mathcal{S}= \big\{ x:\;  x \leq  \gamma(x) \leq x^C  \big\}$.
 If
$R_{L}'(y^C)R_F'(x^C)>1/2$, then $\gamma'(x^C)>0$. This in turn implies the existence of $x<x^C$ such that $x<\gamma(x)<x^C$. Such an $x$ belongs to $\mathcal{S}$,  so Lemma \ref{lemmapouvoirachat} enables us to conclude that $\underline{U}<U(x^C)$. 
\end{proof}

\end{document}